\documentclass[proceedings]{stacs}
\stacsheading{year}{numbers}{city}
\usepackage{subfigure}

\newcommand{\SHARPP}{{\#\rm{P}}}
\newcommand{\Arg}{{\mathrm{Arg}}}
\newcommand{\ptime}{{\rm{P}}}

\theoremstyle{plain}
\theoremstyle{definition}

\begin{document}

% \title[short title]{title}
\title{Holant Problems for Regular Graphs with Complex Edge Functions}

% \author[ref]{Short author}{Author}
\author[lab1]{M. Kowalczyk}{Michael Kowalczyk}
% \address[ref]{Address of authors with ref as reference}
\address[lab1]{Department of Mathematics and Computer Science, Northern Michigan University
  \newline Marquette, MI 49855, USA}  %required
\email{mkowalcz@nmu.edu}  %optional
%\urladdr{http://cs.nmu.edu}  %optional

\author[lab2]{J-Y. Cai}{Jin-Yi Cai}
\address[lab2]{Computer Sciences Department, University of Wisconsin, Madison, WI 53706, USA}	%optional
\email{jyc@cs.wisc.edu}  %optional
%\urladdr{http://cs.wisc.edu}  %optional

\thanks{The second author is supported by NSF CCF-0830488 and CCF-0914969.}

%% mandatory lists of keywords and classifications:
% \titlecomment{OPTIONAL comment concerning the title, \eg, if a variant
% or an extended abstract of the paper has appeared elsewehere}
%%%%%%%%%%%%%%%%%%%%%%%%%%%%%%%%%%%%%%%%%%%%%%%%%%%%%%%%%%%%%%%%%%%%%%%%%%%

%% the abstract has to PRECEDE the command \maketitle:
%% be sure not to issue the \maketitle command twice!

\begin{abstract}
	\noindent We prove a complexity dichotomy theorem for
	Holant Problems on $3$-regular graphs
	with an arbitrary complex-valued edge function.
	Three new techniques are introduced:
	(1) higher dimensional iterations in interpolation;
	(2) Eigenvalue Shifted Pairs, which allow us to 
	prove that a pair  of combinatorial
	gadgets \emph{in combination} succeed in proving \#P-hardness;
	and (3) algebraic symmetrization, which significantly
	lowers the \emph{symbolic complexity} of the proof for
	computational complexity.
	With \emph{holographic reductions} the classification
	theorem also applies to problems beyond the basic model.
\end{abstract}

\maketitle

%%%%%%%%% Begin Paper %%%%%%%%%

%%%%%%%%% Begin Intro %%%%%%%%%

\section{Introduction}\label{section1}
In this paper we consider the following subclass of
Holant Problems~\cite{FOCS08,TAMC}.
An input regular graph  $G = (V, E)$ is given, where
every  $e \in E$ is labeled with a (symmetric) edge function $g$.
The function $g$ takes 0-1 inputs from its incident
nodes and outputs arbitrary values
in $\mathbb{C}$. 
The problem is to compute the quantity
${\rm Holant}(G) = \sum_{\sigma: V \rightarrow \{0,1\}}
\prod_{\{u,v\} \in E} g (\{\sigma(u), \sigma(v)\})$.

Holant Problems are a natural class of
counting problems. As introduced 
in~\cite{FOCS08,TAMC}, the general  Holant Problem framework
can encode all Counting
Constraint Satisfaction Problems (\#CSP). This
includes special cases such as weighted {\sc Vertex Cover},
{\sc Graph Colorings},  {\sc Matchings}, and {\sc Perfect
Matchings}. The subclass of Holant Problems
in this paper can also be considered as (weighted)
$H$-homomorphism (or $H$-coloring) 
problems~\cite{BulatovG05,Homomorphisms,acyclic,DyerG00,Goldberg-4,Hell}
with an arbitrary
$2 \times 2$ symmetric complex matrix $H$,
however \emph{restricted to} regular graphs $G$ as input.
E.g.,  {\sc Vertex Cover} is the case when
$H = 
{
\left [
\begin{array}{cc}
0 & 1 \\
1 & 1
\end{array}
\right ]
}
$. 
When the matrix $H$ is a 0-1 matrix, it is called unweighted.
Dichotomy theorems (i.e., the problem
is either in $\ptime$ or \#P-hard, depending
on $H$)
 for unweighted $H$-homomorphisms with undirected graphs
$H$ and directed acyclic graphs $H$ are given in \cite{DyerG00} and
\cite{acyclic} respectively.
A dichotomy theorem for any symmetric matrix $H$
with non-negative real entries is proved in~\cite{BulatovG05}.
Goldberg et~al.~\cite{Goldberg-4} 
proved a dichotomy theorem for all real symmetric matrices $H$.
Finally, Cai, Chen, and Lu have proved a dichotomy theorem for all 
complex symmetric matrices $H$~\cite{Homomorphisms}.

The crucial difference between Holant Problems and \#CSP 
is that in \#CSP, {\sc Equality} functions of arbitrary
arity are \emph{presumed} to be present. In terms of
$H$-homomorphism problems, this means that the input graph is allowed to
have vertices of arbitrarily high degrees.  This may appear
to be a minor distinction; in fact it has a major impact
on complexity. It turns out that if  {\sc Equality} gates 
of arbitrary arity are
freely available in possible inputs  then
it is technically easier to prove \#P-hardness. 
Proofs of previous
dichotomy theorems make extensive use of constructions called
thickening and stretching. These
constructions require the availability of {\sc Equality} gates of
arbitrary arity (equivalently, vertices of arbitrarily high degrees)
 to carry out.
Proving
 \#P-hardness becomes more
challenging in the degree restricted case.
 Furthermore there are indeed cases
within this class of counting problems where the problem
is \#P-hard for general graphs, but
solvable in $\ptime$ when restricted to 3-regular graphs.

We denote the (symmetric) edge function $g$ by
$[x, y, z]$, where $x = g(0,0)$, $y = g(0,1) = g(1,0)$ and $z = g(1,1)$.
Functions will also be called gates or signatures.
(For  {\sc Vertex Cover}, the function corresponding to $H$
is the {\sc Or} gate,
and is denoted by the signature $[0,1,1]$.)
In this paper we give a dichotomy
theorem for the complexity of Holant Problems
on 3-regular
graphs with arbitrary signature
 $g = [x, y, z]$, where $x, y, z \in \mathbb{C}$.
First, if $y=0$, the Holant Problem is easily solvable
in $\ptime$. Assuming $y \not = 0$ we may normalize $g$ and assume $y=1$.
Our main theorem is as follows:
\begin{theorem}\label{thm:main-intro}
        Suppose $a, b \in \mathbb{C}$, and let $X = a b$, 
$Z = (\frac{a^3 + b^3}{2})^2$.  Then the Holant Problem 
on 3-regular graphs
with $g=[a,1,b]$ is \SHARPP-hard except in the following cases, 
for which the problem is in $\ptime$.
        \begin{enumerate}
                \item $X=1$.
                \item $X=Z=0$.
                \item $X = -1$ and $Z = 0$.
                \item $X = -1$ and $Z = -1$.
        \end{enumerate}
If we restrict the input to planar 3-regular graphs, 
then these four categories are solvable in $\ptime$, 
as well as a fifth category $X^3 = Z$. The problem 
remains \SHARPP-hard in all other cases.
\footnote{Technically, computational complexity involving complex or
real numbers should,  in the Turing model,
be restricted to computable numbers. In other models such as
the Blum-Shub-Smale model~\cite{BSS} no such restrictions are
needed. Our results are not
sensitive to the exact model of computation.}
\qed
\end{theorem}

These results can be extended to $k$-regular graphs (we detail how this
is accomplished in a forthcoming work).
One can also use holographic reductions~\cite{HA_FOCS}
to extend this theorem to more general Holant Problems.

In order to achieve this result, some new proof techniques are introduced.
To discuss this we first take a look at some previous results.
Valiant~\cite{Valiant79b,Valiant:sharpP}
 introduced the powerful technique of \emph{interpolation},
which was further developed by many others.
In \cite{FOCS08} a dichotomy theorem is proved for the case when $g$ is
a Boolean function.  
The technique from  \cite{FOCS08} is to provide 
certain algebraic criteria which ensure that
\emph{interpolation} succeeds, and then apply these
criteria to prove that
(a large number yet) finitely many individual problems 
are \SHARPP-hard.  This involves (a small number of)
   gadget constructions,
and the algebraic criteria are powerful
enough to show that they succeed in  each case.
Nonetheless this involves a case-by-case verification.
In~\cite{TAMC}
this theorem is extended to all real-valued $a$ and $b$,
and we have to deal with infinitely
many problems. So instead of focusing on one problem,
 we devised (a large  number of)
recursive gadgets and analyzed the regions  of $(a,b) \in \mathbb{R}^2$
where they fail to prove \SHARPP-hardness. 
The algebraic criteria from \cite{FOCS08}
are not suitable (Galois theoretic) for general  $a$ and $b$,
and so we formulated weaker but simpler criteria.
Using these criteria, the analysis of the failure set
becomes expressible as containment of semi-algebraic sets.
As semi-algebraic sets are decidable, this offers
the ultimate possibility that \emph{if} we found enough
gadgets to prove \SHARPP-hardness, \emph{then} there is a \emph{computational}
proof (of computational intractability)
 in a finite number of steps.  However this turned out
to be a tremendous undertaking in symbolic
computation, and many
additional ideas were needed to finally carry out this plan.
In particular, it would seem hopeless to extend
that approach to all complex  $a$ and $b$.

In this paper, we introduce three new ideas.
(1) We introduce a method to construct gadgets that
carry out iterations at a higher dimension, and then collapse
to a lower dimension for the purpose of constructing
unary signatures. This involves
a starter gadget,  a recursive iteration gadget, and 
a finisher gadget.  We prove a lemma that guarantees that
among polynomially many iterations, some subset of them
satisfies properties sufficient for interpolation to succeed
(it may not be known \emph{a priori} which subset worked, but
that does not matter).  
(2) Eigenvalue Shifted Pairs are coupled pairs of gadgets
whose transition matrices differ by $\delta I$ where $\delta \ne 0$.
They have shifted eigenvalues, and by analyzing their failure
conditions, we can show that except on very rare points,
one or the other gadget succeeds.
(3) Algebraic symmetrization.  We derive a new expression of the Holant
polynomial over 3-regular graphs, with  a  crucially reduced degree.
This simplification of the Holant and related polynomials condenses the
problem of proving $\SHARPP$-hardness to the point where 
all remaining cases can be handled by symbolic computation.
We also use the same expression to prove tractability.

The rest of this paper is organized as follows.  In Section
\ref{background} we discuss notation and background information.
In Section \ref{interpolation} we cover interpolation techniques,
including how to collapse higher dimensional iterations to
interpolate unary signatures.  In Section \ref{complexSignatures}
we show how to perform algebraic symmetrization of the Holant,
and introduce Eigenvalue Shifted Pairs (ESP) of gadgets.
Then we combine the new techniques
to prove Theorem~\ref{thm:main-intro}.

%%%%%%%%% Begin Background %%%%%%%%%

\section{Notations and Background}\label{background}

We state the counting framework more formally.
A {\it signature grid} $\Omega = (G, {\mathcal F}, \pi)$
consists of a labeled graph $G=(V,E)$ where $\pi$ labels
each vertex $v \in V$  with
a function $f_v \in {\mathcal F}$.
We consider all edge assignments $\xi: E \rightarrow \{0,1\}$;
$f_v$ takes inputs
from its incident edges $E(v)$ at $v$
and outputs values in $\mathbb{C}$. 
The counting problem on the instance $\Omega$
is to compute\footnote{The term
Holant was first introduced by
Valiant in \cite{HA_FOCS}
to denote a related exponential sum.} 
\[{\rm Holant}_\Omega=\sum_{\xi}
\prod_{v\in V} f_v(\xi\mid_{E(v)}).\]

Suppose $G$ is a bipartite graph $(U, V, E)$
such that each $u \in U$ has degree 2.
Furthermore suppose each $v \in V$ is labeled by
an {\sc Equality} gate $=_k$ where $k = {\rm deg}(v)$.
Then any non-zero term in ${\rm Holant}_\Omega$ 
corresponds to a 0-1 assignment $\sigma: V \rightarrow \{0,1\}$.
In fact, we can merge the two incident edges at $u \in U$
 into one edge $e_u$, and label this edge $e_u$
 by the function $f_u$.
This gives an edge-labeled graph  $(V, E')$
where $E' = \{e_u : u \in U\}$.
For an edge-labeled graph  $(V, E')$ where $e \in E'$ has label $g_e$,
 ${\rm Holant}_\Omega = \sum_{\sigma: V \rightarrow \{0,1\}}
\prod_{e = (v,w) \in E'}  g_e (\sigma(v), \sigma(w))$.
If each $g_e$ is the same function $g$ (but assignments
$\sigma: V \rightarrow [q]$ take values in a finite set $[q]$)   this
is exactly the $H$-coloring problem (for undirected graphs
$g$ is a symmetric function). 
In particular, if $(U, V, E)$ is a $(2,k)$-regular bipartite graph,
equivalently $G' = (V, E')$ is a  $k$-regular graph, then this is
the $H$-coloring problem restricted to $k$-regular graphs.
In this paper we will 
discuss 3-regular graphs,
where each $g_e$ is the same symmetric complex-valued function.
We also remark that
for general bipartite graphs $(U, V, E)$,   giving {\sc Equality} (of various
arities)
to all vertices on one side $V$ defines  \#CSP 
as a special case of Holant Problems.
But whether {\sc Equality} of
various arities are present has a major 
impact on complexity, thus Holant Problems are a refinement of
\#CSP.

A symmetric function $g: \{0,1\}^k \rightarrow \mathbb{C}$ can be 
denoted as
$[g_0, g_1,
\ldots, g_k]$, where $g_i$ is the value of $g$
on inputs of Hamming weight $i$.
They are also called {\it signatures}.
Frequently we will revert back to the bipartite view:
for $(2,3)$-regular bipartite graphs $(U, V, E)$,
if every $u \in U$ is labeled $g= [g_0, g_1, g_2]$
and every $v \in V$ is labeled $r =[r_0, r_1, r_2, r_3]$, then
we also use $\#[g_0, g_1, g_2] \mid [r_0, r_1, r_2, r_3]$
to denote the Holant Problem.
 Note that $[1,0,1]$ and $[1,0,0,1]$
are {\sc Equality} gates $=_2$ and $=_3$ respectively,
and the main dichotomy theorem in this paper
is about $\#[x,y,z] \mid [1,0,0,1]$, for all $x,y,z \in \mathbb{C}$.
We will also denote $\mathrm{Hol}(a,b) = \#[a,1,b] \mid [1,0,0,1]$.
More generally,
If ${\mathcal G}$ and ${\mathcal R}$ are sets of signatures,
and vertices of  $U$ (resp. $V$)
 are labeled by signatures from ${\mathcal G}$ (resp. ${\mathcal R}$),
then we also use $\# {\mathcal G} \mid {\mathcal R}$ to denote 
the bipartite Holant Problem.
Signatures in ${\mathcal G}$ are called {\it generators}
and signatures in ${\mathcal R}$ are called {\it recognizers}.
This notation is 
particularly convenient when we perform holographic
transformations.
Throughout this paper, all $(2,3)$-regular bipartite graphs 
are arranged with generators on the degree 2 side and 
recognizers on the degree 3 side.

We use $\Arg$ to denote the principal value of the complex argument; i.e.,
$\Arg(c) \in (-\pi, \pi]$ for all nonzero $c \in \mathbb{C}$.

\subsection{${\mathcal F}$-Gate}
Any signature from ${\mathcal F}$ is available at a vertex
as part of an input graph.
 Instead of a single vertex, we can use graph fragments to generalize 
this notion. An ${\mathcal F}$-gate $\Gamma$ is a pair $(H, {\mathcal F})$,
where $H = (V,E,D)$ is a graph with some dangling edges $D$ 
(Figure \ref{exampleGadgets} contains some examples).
Other than these dangling edges,
an ${\mathcal F}$-gate is the same as a signature grid. The role of
dangling edges is similar to that of external nodes in Valiant's
notion~\cite{Valiant:Qciricuit}, however
we allow more than one dangling edge for a node.
In $H=(V,E,D)$ each node is assigned a
function in ${\mathcal F}$
(we do not consider ``dangling'' leaf nodes at the end of a dangling
edge among these), $E$ are the regular edges, and $D$ are the 
dangling edges. Then we can define a function for this 
${\mathcal F}$-gate $\Gamma = (H, {\mathcal F})$,
\[\Gamma(y_1, y_2,
\ldots, y_q)=\sum_{(x_1, x_2, \ldots, x_p) \in \{0,1\}^p}H(x_1, x_2, \ldots,
x_p, y_1, y_2, \ldots, y_q),\] where 
$p=|E|$, $q=|D|$, 
$(y_1, y_2, \ldots, y_q) \in
\{0,1\}^q$ denotes an assignment on the dangling edges, and
$H(x_1, x_2, \ldots, x_p, y_1, y_2, \ldots, y_q)$ denotes
the value of the partial signature grid on an assignment of all edges,
i.e., the product of evaluations at every vertex of $H$,
for $(x_1, x_2, \ldots, x_p, y_1, y_2, \ldots, y_q) \in
\{0,1\}^{p+q}$.
\begin{figure}[h!tb]
\centering
\subfigure[A starter gadget]{
  \qquad\includegraphics[scale=.65]{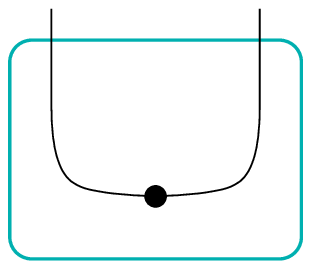}\qquad
  \label{starterGadget}
}
\subfigure[A recursive gadget]{
  \qquad\includegraphics[scale=.65]{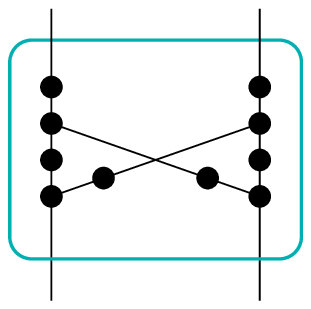}\qquad
  \label{recursiveGadget}
}
\subfigure[A finisher gadget]{
  \qquad\includegraphics[scale=.65]{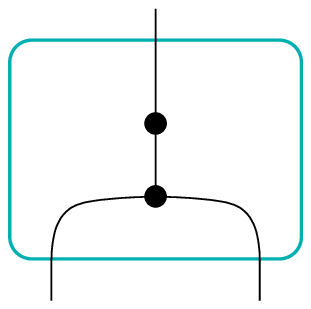}\qquad
  \label{finisherGadget}
}
\subfigure[A planar embedding of a single iteration]{
  \qquad\includegraphics[scale=.65]{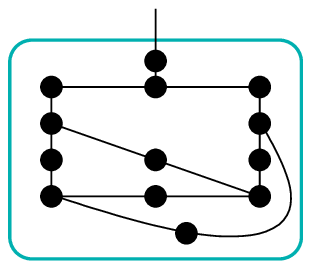}\qquad
  \label{iteration}
}
\caption{Examples of binary starter, recursive, and finisher gadgets}
\label{exampleGadgets}
\end{figure}
We
will also call this function the signature of the ${\mathcal F}$-gate
$\Gamma$. 
An ${\mathcal F}$-gate can be used in a signature grid as if
it is just a single node with the same signature. We note that
even for a very simple signature set ${\mathcal F}$, the signatures for
all ${\mathcal F}$-gates can be quite complicated and expressive.
Matchgate signatures are an example~\cite{Valiant:Qciricuit}.

The dangling edges of an $\mathcal{F}$-gate are
considered as input or output variables.
Any $m$-input $n$-output $\mathcal{F}$-gate can be 
viewed as a $2^n$ by $2^m$ matrix $M$ which transforms arity-$m$ 
signatures into arity-$n$ signatures
(this is true even if $m$ or $n$ are 0). 
Our construction will transform symmetric signatures to 
symmetric signatures. This implies that there exists
 an equivalent $n+1$ by $m+1$ matrix $\widetilde{M}$
which operates directly on column vectors written 
in symmetric signature notation.
We will henceforth identify the matrix $\widetilde{M}$ with the 
$\mathcal{F}$-gate itself.  The constructions in this paper are based upon
three different types of bipartite $\mathcal{F}$-gates which we call 
{\em starter gadgets}, {\em recursive gadgets}, and {\em finisher gadgets}.
An {\em arity-$r$ starter gadget} is an $\mathcal{F}$-gate with no input
but $r$ output edges.
If an $\mathcal{F}$-gate has $r$ input and $r$ output edges
then it is called an {\em arity-$r$ 
recursive gadget}.  Finally, an $\mathcal{F}$-gate is an {\em arity-$r$
finisher gadget} if it has $r$ input edges
1 output edge.
As a matter of convention, we consider any dangling edge incident 
with a generator as an output edge and any dangling edge incident
with a recognizer as an input edge; see Figure \ref{exampleGadgets}.

%%%%%%%%% Begin Interpolation %%%%%%%%%

\section{Interpolation Techniques} \label{interpolation}

\subsection{Binary recursive construction}
In this section, we develop our new technique of higher dimensional
iterations  for interpolation of unary signatures.

\begin{lemma}\label{linAlg}
        Suppose $M \in \mathbb{C}^{3 \times 3}$
is a nonsingular matrix, $s \in \mathbb{C}^{3}$ is a 
nonzero vector, and for all integers $k \ge 1$, $s$ is not 
a column eigenvector of $M^k$.  Let $F_i \in 
\mathbb{C}^{2 \times 3}$ be three matrices, 
where ${\rm rank}(F_i) = 2$ for $1 \le i \le 3$, and the intersection
of the row spaces of $F_i$ is trivial $\{0\}$. 
Then for every $n$,
there exists some $F \in \{F_i : 1 \le i \le 3\}$,
and some $S \subseteq \{F M^k s : 0 \leq k \leq n^3\}$, 
such that $|S| \ge n$ and vectors in $S$ are 
\emph{pairwise} linearly independent.
\end{lemma}\label{one-finisher-works}
\proof
Let $k > j \ge 0$ be integers. 
Then $M^k s$ and $M^j s$ are nonzero and also linearly independent, since
otherwise $s$ is an eigenvector of $M^{k-j}$.
Let $N = [M^j s, M^k s] \in \mathbb{C}^{3 \times 2}$,
then ${\rm rank}(N) =2$, and 
$\mathrm{ker}(N^\mathrm{T})$ is a 1-dimensional linear subspace.
It follows that there exists an $F \in \{F_i : 1 \le i \le 3\}$
such that the row space of $F$ does not contain
$\mathrm{ker}(N^\mathrm{T})$, and hence has trivial
intersection with $\mathrm{ker}(N^\mathrm{T})$.
In other words, 
$\mathrm{ker}(N^\mathrm{T}F^\mathrm{T}) = \{0\}$.
We conclude that $FN \in \mathbb{C}^{2 \times 2}$
has rank 2, and $FM^j s$ and $FM^k s$ are linearly independent.

Each $F_i$, where $1 \le i \le 3$, defines a coloring of the set 
$K = \{0, 1, \dots, n^3\}$ as follows: color $k \in  K$ with the 
linear subspace spanned by $F_i M^k s$. Thus, $F_i$ defines an 
equivalence relation $\approx_i$ where $k \approx_i k'$ 
iff they receive the same color.  Assume for a contradiction 
that for each $F_i$, where $1 \le i \le 3$, there are not $n$ pairwise 
linearly independent vectors among $\{F_i M^k s : k \in  K\}$.  
Then, including possibly the 0-dimensional space $\{0\}$, 
there can be at most $n$ distinct colors assigned by $F_i$.  
By the pigeonhole principle, some $k$ and $k'$ with $0 \leq k < k' \leq n^3$ 
must receive the same color for all $F_i$, where $1 \le i \le 3$. 
This is a contradiction and we are done.
\qed

The next lemma says that under suitable conditions we can construct
all unary signatures $[x, y]$. The method will be interpolation
at a higher dimensional iteration, and finishing up with a suitable
\emph{finisher} gadget.  The crucial new technique here is
that when iterating at a higher dimension, we can guarantee
the existence of \emph{one} finisher gadget that succeeds on 
polynomially many steps, which results in overall success.
Different finisher gadgets may work for different initial signatures
and different input size $n$, but these need not be known in advance and 
have no impact on the final success of the reduction.

\begin{lemma}\label{mainLemma} 
        Suppose that the following gadgets can be built using complex-valued
        signatures from a finite generator set $\mathcal{G}$ and a finite 
        recognizer set $\mathcal{R}$.
        \begin{enumerate}
		\item A binary starter gadget with nonzero signature $[z_0, z_1, z_2]$.
		\item A binary recursive gadget with nonsingular recurrence matrix $M$, 
		for which $[z_0,z_1,z_2]^\mathrm{T}$ is not a column eigenvector of 
		$M^k$ for any positive integer $k$.
		\item Three binary finisher gadgets with rank 2 matrices $F_1, F_2, F_3
		 \in \mathbb{C}^{2 \times 3}$, where the intersection of the 
		row spaces of $F_1$, $F_2$, and $F_3$ is the zero vector.
		\label{finisherRequirement}
	\end{enumerate}
	Then for any $x, y  \in \mathbb{C}$, $\#\mathcal{G} \cup \{[x, y]\} \mid 
	\mathcal{R} \leq_T \#\mathcal{G} \mid \mathcal{R}$.
\end{lemma}

\proof
	The construction begins with the binary starter gadget with signature
	$[z_0, z_1, z_2]$, which we call $N_0$.  Let $\mathcal{F} = \mathcal{G}
	\cup \mathcal{R}$.  Recursively, $\mathcal{F}$-gate $N_{k+1}$ is 
	defined to be $N_k$ connected to the binary recursive gadget in such 
	a way that the input edges of the binary recursive gadget are 
	merged with the output edges of $N_k$.  Then $\mathcal{F}$-gate 
	$G_k$ is defined to be $N_k$ connected to one of the finisher gadgets,
	with the input edges of the finisher gadget merged with the output
	edges of $N_k$ (see Figure \ref{iteration}).  Herein we
	analyze the construction with respect to a given bipartite signature
	grid $\Omega$ for the Holant Problem 
$\#\mathcal{G} \cup \{[x, y]\} \mid \mathcal{R}$, with underlying graph
	$G = (V, E)$.
  Let $Q \subseteq V$ be the
	set of vertices with $[x,y]$ signatures, and let $n = |Q|$.
	By Lemma \ref{linAlg} fix $j$ so that at least $n+2$ of the first
	$(n+2)^3 + 1$ vectors of the form $F_j M^k [z_0,z_1,z_2]^\mathrm{T}$
	are pairwise linearly independent.  We use finisher gadget $F_j$ in
	the recursive construction, so that the signature of $G_k$ is 
	$F_j M^k [z_0,z_1,z_2]^\mathrm{T}$, which we denote by $[X_k, Y_k]$.
	We note that there exists a subset $S$ of these signatures for which
	each $Y_k$ is nonzero and $|S| = n + 1$.  We will argue using only
	the existence of $S$, so there is no need to algorithmically ``find''
	such a set, and for that matter, one can try out all three finisher
	gadgets without any need to determine which finisher gadget is 
	``the correct one'' beforehand.  If we replace every element of 
	$Q$ with a copy of $G_k$, we obtain an instance of 
	$\#\mathcal{G} \mid \mathcal{R}$ (note that the correct bipartite 
	signature structure is preserved), and we denote this new signature 
	grid by $\Omega_k$.  Then
	\begin{eqnarray*}
		\mathrm{Holant}_{\Omega_k} 
			= \sum_{0\leq i \leq n} c_i X_k^i Y_k^{n-i}
	\end{eqnarray*}
	where $c_i = \sum_{\sigma \in J_i} \prod_{v \in V \setminus Q} 
	f_v(\sigma |_{E(v)})$, $J_i$ is the set of $\{0,1\}$ edge
	assignments where the number of 0s assigned to the edges incident 
	to the copies of $G_k$ is $i$, $f_v$ is the signature at $v$,
	and $E(v)$ is the set of edges incident to $v$.
	The important point is that
	the $c_i$ values do not depend on $X_k$ or $Y_k$.  Since each 
	signature grid $\Omega_k$ is an instance of $\#\mathcal{G} \mid 
	\mathcal{R}$, $\mathrm{Holant}_{\Omega_k}$ can be solved exactly 
	using the oracle.  Carrying out this process for every $k \in 
	\{0, 1, \dots, (n+2)^3\}$, we arrive at a linear system where the 
	$c_i$ values are the unknowns.
	\begin{eqnarray*}
		\left[\begin{array}{c} \mathrm{Holant}_{\Omega_{0}} \\ 
		\mathrm{Holant}_{\Omega_{1}} \\ \vdots \\ 
		\mathrm{Holant}_{\Omega_{(n+2)^3}} \end{array} \right] &=& 
		\left[\begin{array}{cccc}
			X_{0}^0 Y_{0}^n & X_{0}^1 Y_{0}^{n-1} & 
			\cdots & X_{0}^n Y_{0}^{0} \\
			X_{1}^0 Y_{1}^n & X_{1}^1 Y_{1}^{n-1} & 
			\cdots & X_{1}^n Y_{1}^{0} \\
			\vdots & \vdots & \ddots & \vdots \\
			X_{(n+2)^3}^0 Y_{(n+2)^3}^n & 
			X_{(n+2)^3}^1 Y_{(n+2)^3}^{n-1} & 
			\cdots & X_{(n+2)^3}^n Y_{(n+2)^3}^{0} \\
		\end{array}\right]
		\left[\begin{array}{c} c_0 \\ c_1 \\ \vdots \\ 
		c_n \end{array}\right].
	\end{eqnarray*}
	For $0 \leq i \leq n$, let $k_i$ such that
	$S = \{[X_{k_0}, Y_{k_0}], [X_{k_1}, Y_{k_1}], \dots, [X_{k_n}, Y_{k_n}]\}$, 
	and let $[x_i, y_i] = [X_{k_i}, Y_{k_i}]$.  Then we have a subsystem
	\begin{eqnarray*}
		\left[\begin{array}{c} y_{0}^{-n} \cdot 
		\mathrm{Holant}_{\Omega_{k_0}} \\ y_{1}^{-n} \cdot 
		\mathrm{Holant}_{\Omega_{k_1}} \\ \vdots \\ y_{n}^{-n} 
		\cdot \mathrm{Holant}_{\Omega_{k_n}} \end{array} \right] &=& 
		\left[\begin{array}{cccc}
			x_{0}^0 y_{0}^0 & x_{0}^1 y_{0}^{-1} & \cdots 
			& x_{0}^n y_{0}^{-n} \\
			x_{1}^0 y_{1}^0 & x_{1}^1 y_{1}^{-1} & \cdots 
			& x_{1}^n y_{1}^{-n} \\
			\vdots & \vdots & \ddots & \vdots \\
			x_{n}^0 y_{n}^0 & x_{n}^1 y_{n}^{-1} & \cdots 
			& x_{n}^n y_{n}^{-n} \\
		\end{array}\right]
		\left[\begin{array}{c} c_0 \\ c_1 \\ \vdots \\ 
		c_n \end{array}\right].
	\end{eqnarray*}
	The matrix above has entry $(x_{r} / y_{r})^{c}$ at index $(r,c)$.  
	Due to pairwise linear independence of $[x_{r}, y_{r}]$,  
	$x_{r} / y_{r}$ is pairwise distinct for $0 \leq r \leq n$.
	Hence this is a Vandermonde system of full rank.
	Therefore the initial feasible linear system has full rank and 
	we can solve it for the $c_i$ values.  With these values in hand,
	we can calculate $\mathrm{Holant}_\Omega = \sum_{0\leq i \leq n}
	c_i x^i y^{n-i}$ directly, completing the reduction.
\qed

The ability to simulate all unary signatures will allow us
to prove $\SHARPP$-hardness.
The next lemma says that, if $\mathcal{R}$ contains the
{\sc Equality} gate $=_3$, then
 other than on a 1-dimensional curve $ab=1$ and
an isolated point $(a,b)=(0,0)$, the ability to simulate unary signatures
gives a reduction from {\sc Vertex Cover}.
Note that counting {\sc Vertex Cover} on 3-regular
graphs is just $\#[0,1,1] \mid [1,0,0,1]$.
Xia et al. showed that this is   $\SHARPP$-hard even 
when the input is restricted to 
3-regular planar graphs~\cite{XiaZZ07}.
We will see shortly that
on the curve $ab=1$ and at $(a,b)=(0,0)$,
the problem  $\mathrm{Hol}(a,b)$ is tractable.

\begin{lemma}\label{mainLemma2} 

Suppose that $(a, b) \in \mathbb{C}^2 - \{(a,b) : ab = 1\} -
\{(0,0)\}$ and let $\mathcal{G}$ and $\mathcal{R}$ be finite signature sets 
where $[a,1,b] \in \mathcal{G}$ and $[1,0,0,1] \in \mathcal{R}$.  Further 
assume that $\#\mathcal{G} \cup \{[x_i, y_i] : 0 \leq i < m\} \mid \mathcal{R}
\leq_T \#\mathcal{G} \mid \mathcal{R}$ for any $x_i, y_i  \in \mathbb{C}$ and
$m \in \mathbb{Z}^+$.
Then $\#\mathcal{G} \cup \{[0, 1, 1]\} \mid \mathcal{R}
\leq_T \#\mathcal{G} \mid \mathcal{R}$, and
 $\#\mathcal{G} \mid \mathcal{R}$ is $\SHARPP$-hard.

\end{lemma} 

\proof
	Assume $ab \ne 1$ and $(a,b) \ne (0,0)$.
	Since $\mathrm{Hol}(0,1)$ (which is
	the same as $\#[0,1,1] \mid [1,0,0,1]$, or counting
	vertex covers on 3-regular graphs) is 
	$\SHARPP$-hard, we only need to show how to simulate the 
	generator signature $[0,1,1]$.  We split this into three
	cases, and use a chain of three reductions, each involving
	a gadget in Figure \ref{vcGadgets} (each gadget has
	$[1,0,0,1]$ assigned to the degree 3 vertices
	and $[a,1,b]$ assigned to the degree 2 vertices).
	\begin{enumerate}
		\item $ab \ne 0$ and $ab \ne -1$
		\item $ab = 0$
		\item $ab = -1$
	\end{enumerate}

	If $ab \ne 0$ and $ab \ne -1$, then we use
	Gadget 3, and we set its unary signatures
	to be
	$\theta = [(a b + 1)/(1 - a b), -a^2 (a b + 1)/(1 - a b)]$, 
	$\gamma = [-a^{-2}, b^{-1}(1 + a b)^{-1}]$, and
	$\rho= [-b/(a b - 1), a/(a b - 1)]$.
	Calculating the resulting signature of Gadget 3, we find
	that it is $[0,1,1]$ as desired.

	If $ab = 0$ then assume without loss of generality that $a=0$ 
	and $b \ne 0$.  This time we use Gadget 1, setting
	$\theta = [b, b^{-1}]$.  Then Gadget 1 simulates a
	$[b^{-1}, 1, 2b]$ generator signature, but since this 
	signature fits
	the criteria of case 1 above, we are done by reduction
	from that case.

	Similarly, if $ab=-1$, then Gadget 2 exhibits 
	a generator signature of the form
	$[0,1,5/(2a)]$ under the signatures 
	$\theta = [1/(6 a), -a/24]$ and
	$\gamma = [-3/a, a]$.  Since $5/(2a)$ is
	nonzero, we are done by reduction from case 2.
\qed
\begin{figure}[h!tb]
\centering
\subfigure[Gadget 1]{
  \includegraphics[scale=.6]{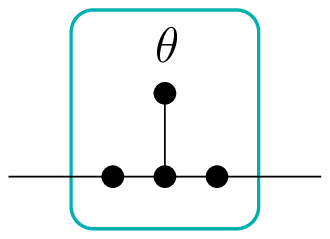}
  \label{gadget1}
}
\subfigure[Gadget 2]{
  \includegraphics[scale=.6]{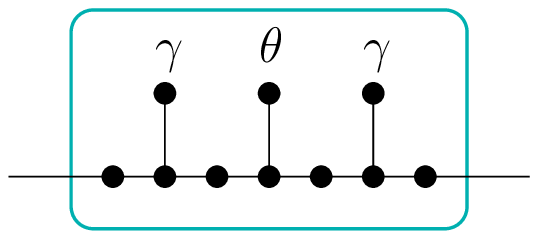}
  \label{gadget2}
}
\subfigure[Gadget 3]{
  \includegraphics[scale=.6]{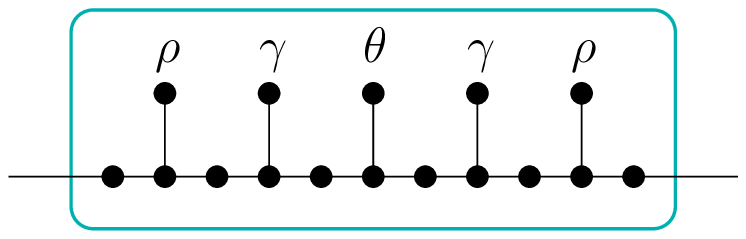}
  \label{gadget3}
}
\caption{Gadgets used to simulate the [0,1,1] signature}
\label{vcGadgets}
\end{figure}
It will be helpful to have conditions that are easier to check than those in Lemma \ref{mainLemma}.  To this end, we establish condition 2 in terms of eigenvalues, and we build general-purpose finisher gadgets to eliminate condition 3.  Let $M_4$, $M_5$, and $F$ be the recurrence matrices for Gadget 4, Gadget 5, and the simplest possible binary finisher gadget (each built using generator signature $[a,1,b]$ and recognizer signature $[1,0,0,1]$; see Figures \ref{gadget4}, \ref{gadget5}, and \ref{finisherGadget}).  Provided that $ab \ne 1$ and $a^3 \ne b^3$, it turns out that the finisher gadget sets $\{F, F M_4, F M_4^2\}$ and $\{F, F M_4, F M_5\}$ satisfy condition 3 of Lemma \ref{mainLemma} when $ab \ne 0$ and $ab = 0$, respectively.  Together with Lemma \ref{mainLemma2}, these observations yield the following.

\begin{theorem}\label{mainBinary}
        If the following gadgets can be built using generator
        $[a,1,b]$ and recognizer $[1,0,0,1]$ where 
        $a, b \in \mathbb{C}$, 
        $ab \ne 1$, and $a^3 \ne b^3$, then the problem $\mathrm{Hol}(a,b)$ is $\SHARPP$-hard.
        \begin{enumerate}
        	\item A binary recursive gadget with nonsingular recurrence 
        	matrix $M$ which has eigenvalues $\alpha$ and $\beta$ such 
        	that $\frac{\alpha}{\beta}$ is not a root of unity. 
        	
                \item A binary starter gadget with signature $s$ which 
                is not orthogonal to any row eigenvector of $M$.
        \end{enumerate}
\end{theorem}
\proof
	First we show how to build general-purpose binary finisher gadgets
	for the main construction using the assumed generator 
	and recognizer, starting first with the case where $ab \ne 0$.  Using
	the simplest possible choice for a finisher gadget $F$ (Figure
	\ref{finisherGadget}), we get $F = \left[\begin{array}{ccc}a & 0 &
	1 \\ 1 & 0 & b \end{array}\right]$.  Let $M_4$ be the recurrence 
	matrix for binary recursive Gadget 4 (Figure \ref{gadget4}),
	and we calculate that
	\begin{eqnarray*}
		M_4 &=& 
		\left[\begin{array}{ccc}
			a^3 & 2 a & b \\
			a^2 & a b + 1 & b^2 \\
			a & 2 b & b^3
		\end{array}\right].
	\end{eqnarray*}
	We build two more finisher 
	gadgets $F'$ and $F''$ using Gadget 4 so that $F' = FM_4$ and 
	$F'' = FM_4^2$.  Since $F$ and $M_4$ both have full rank (note 
	$\mathrm{det}(M_4)=ab(ab-1)^3$), it follows that $F'$ and $F''$ 
	also have full rank.  Now we will show that the row spaces of
	$F$, $F'$ and $F''$ have trivial intersection, and it suffices
	to verify that the cross products of the row vectors of $F$, $F'$,
	and $F''$ (denoted respectively by $v$, $v'$, and $v''$) are linearly independent.  
	(To see this, suppose $u$ is a complex vector in the intersection of the row spaces of $F_1$, $F_1'$, and $F_1''$.  Then $v$, $v'$, and $v''$ are all orthogonal to $u$, but since $v$, $v'$, and $v''$ are linearly independent, they span the conjugate vector $\overline{u}$ which is then also orthogonal to $u$.  This means $|u|^2 = u \overline{u} = 0$, and that $u = 0$.)  
	The cross products of the row vectors of $F$, $F'$,
	and $F''$ are $[0,1-ab,0]$, $(ab-1)^2 [2b^2,-ab(ab+1),2a^2]$,
	and $(ab-1)^3 [2b(a^2b^3+a^2+ab^2+b^4),
	-ab(a^3b^3 + 2a^3 + 2a^2b^2 + ab + 2b^3),2a(a^4 + a^3b^2 + a^2b + b^2)]$ 
	respectively.  Then to see that these 3 vectors are linearly
	independent, it suffices to verify that the matrix 
	$\left[\begin{array}{cc} 2 b^2 & 2 a^2 \\ 
	2b (a^2b^3+a^2+ab^2+b^4) & 2a (a^4 + a^3b^2 + a^2b + b^2) \end{array}\right]$
	is nonsingular.  Since $a \ne 0$, $b \ne 0$, and $ab \ne 1$, we just check
	$\mathrm{det}\left[\begin{array}{cc} b & a \\ 
	a^2b^3 + a^2 + ab^2 + b^4 & a^4 + a^3b^2 + a^2b + b^2
	\end{array}\right] = (ab-1)(a^3-b^3) \ne 0$, so the row spaces of 
	$F$, $F'$, and $F''$ have trivial intersection when $ab \ne 0$.
	
	Now suppose $ab=0$.  Since $a^3 \ne b^3$, by symmetry, if $ab = 0$ we may 
	assume without loss of generality that $a \ne 0$ 
	and $b=0$.  Let $M_5$ be the recurrence matrix for binary
	recursive Gadget 5 (Figure \ref{gadget5}).
	\begin{eqnarray*}
		M_5 &=&
		\left[\begin{array}{ccc}
			a^6 + 2 a^3 + 1 & 2 a^4 + 2 a & a^2 \\
			a^5 + a^2 & 2 a^3 + 1 & a \\
			a^4 & 2 a^2 & 1
		\end{array}\right].
	\end{eqnarray*}
	Composing $F$ with $M_5$, we get a
	finisher gadget with matrix $FM_5$, which has full rank since
	$F$ has full rank and $\mathrm{det}(M_5) = 1$.  It is also
	straightforward to see that $F'$ has full rank, as 
	$F' = 
	\left[\begin{array}{ccc}
		a^4 + a & 2a^2 & 0 \\
		a^3 & 2 a & 0 
	\end{array}\right]
	$.
	The cross products of the rows of $F$, $F'$, and $FM_5$ are 
	$[0,1,0]$, $[0,0,2a^2]$, and $[-2a, 2a^3 + 1, -2a^2(1+a)(a^2-a+1)]$
	respectively.  Then the matrix of cross products 
	is clearly nonsingular, and we conclude that for any $a, b \in
	\mathbb{C}$, we have 3 finisher gadgets satisfying 
	item \ref{finisherRequirement} of Lemma 
	\ref{mainLemma} unless $ab = 1$
	or $a^3 = b^3$.

	Now we want to show that $s$ is not a column eigenvector of $M^k$ for 
	any positive integer $k$ (note that $s$ is nonzero by assumption).
	Writing out the Jordan Normal Form for $M$, we have 
	$M^k s = T^{-1} D^k T s$, where $D^k$ has the form 
	$\left[\begin{array}{ccc} \alpha^k & 0 & 0 \\ 0 & \beta^k & 0 
	\\ 0 & * & * \end{array}\right]$, and where $\alpha$ and
	$\beta$ are eigenvalues of $M$ for which $\frac{\alpha}{\beta}$ is not
	a root of unity.  Let $t = Ts$ and write 
	$t = [c,d,e]^\mathrm{T}$.  By hypothesis, $s$ is not orthogonal 
	to the first two rows of $T$, thus $c, d \ne 0$.  If $s$ were 
	an eigenvector of $M^k$ for some positive integer $k$, then 
	$T^{-1} D^k T s = M^k s = \lambda s$ for some nonzero complex
	value $\lambda$, and $D^k t = T (\lambda s) = \lambda t$.  But 
	then $c \alpha^k =  \lambda c$ and $d \beta^k   =  \lambda  d$,
	which means $\frac{\alpha^k}{\beta^k} = 1$, contradicting the 
	fact that $\frac{\alpha}{\beta}$ is not a root of unity.
	
	We have now met all the criteria for Lemma \ref{mainLemma}, so the
	reduction $\# \mathcal{S} \cup \{[a, 1, b], [x, y]\} \mid 
	[1,0,0,1] \leq_T \# \mathcal{S} \cup \{[a, 1, b]\} \mid [1,0,0,1]$
	holds for any $x,y \in \mathbb{C}$ and any finite signature set
	$\mathcal{S}$.  By Lemma \ref{mainLemma2} 
	the problem $\mathrm{Hol}(a,b)$ is $\SHARPP$-hard.
\qed

\subsection{Unary recursive construction}

Now we consider the unary case.  The following lemma 
arrives from \cite{Vadhan01} and is stated explicitly in \cite{TAMC}.  It 
can be viewed as a unary version of Lemma \ref{mainLemma} without finisher 
gadgets.

\begin{lemma}\label{unaryConstruction}
	Suppose there is a unary recursive gadget with nonsingular matrix 
	$M$ and a unary starter gadget with nonzero signature vector $s$.  
	If the ratio of the eigenvalues of $M$ is not a root of unity and 
	$s$ is not a column eigenvector of $M$, then these gadgets can 
	be used to interpolate all unary signatures.
\qed
\end{lemma}
Surprisingly, a set of general-purpose starter 
gadgets can be made for this construction as long as
$a b \ne 1$ and $a^3 \ne b^3$, so we refine this lemma by
eliminating the starter gadget requirement.  The starter gadgets are
$Fs$, $F M_4 s$, and $F M_6 s$ where $M_6$ is Gadget 6 and $s$ is the 
single-vertex starter gadget (see Figures \ref{gadget6} and \ref{starterGadget}).

\begin{theorem}\label{mainUnary}
	Suppose there is a unary recursive gadget with nonsingular matrix 
	$M$, and the ratio of the eigenvalues of $M$ is not a root of unity.
	Then for any $a, b \in \mathbb{C}$ where $ab \ne 1$ and $a^3 \ne b^3$,
	there is a starter gadget built using generator $[a,1,b]$ and recognizer
	$[1,0,0,1]$ for which the resulting construction
	can be used to interpolate all unary signatures.
\end{theorem}
\proof
	Let $M_4$, $M_6$, $F$, and $s$ be Gadget 4, 
	Gadget 6, the binary finisher 
	gadget in Figure \ref{finisherGadget}, 
	and single-vertex binary starter gadget 
	(Figure \ref{starterGadget}), respectively.
	Note $s = [a,1,b]^\mathrm{T}$ and 
	\begin{eqnarray*}
		M_6 &=& 
		\left[\begin{array}{ccc}
			a^3 + 1 & 0 & a^2 + b \\
			a^2 + b & 0 & a + b^2 \\
			a + b^2 & 0 & b^3 + 1
		\end{array}\right].
	\end{eqnarray*}
	Using Mathematica$^{\rm TM}$, we verify 
	that the block matrices $[F M_4 s\ Fs]$, $[F M_6 s\ Fs]$, and $[F M_4 s\ F M_6 s]$ are
	all nonsingular provided that $a b \ne 1$ and $a^3 \ne b^3$, so
	the vectors $Fs$, $F M_4 s$, and $F M_6 s$ are pairwise linearly 
	independent.  If the ratio of the eigenvalues of $M$ is not a root of unity, 
	then the eigenvalues of $M$ are distinct, each eigenvalue 
	corresponds to an eigenspace of dimension 1, and at least one 
	element of $\{F M_4 s, F M_6 s, Fs\}$ is not a column eigenvector of $M$. 
	The corresponding starter gadget can be used with $M$ in a recursive construction and 
	the result follows from Lemma \ref{unaryConstruction}.
\qed

\begin{figure}[b!ht]
\centering
\subfigure[Gadget 4]{
  \includegraphics[scale=.6]{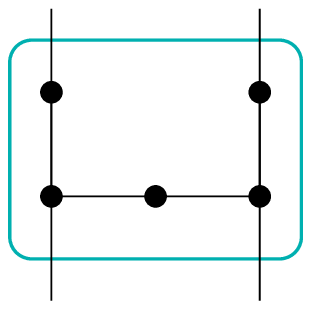}
  \label{gadget4}
}
\subfigure[Gadget 5]{
  \includegraphics[scale=.6]{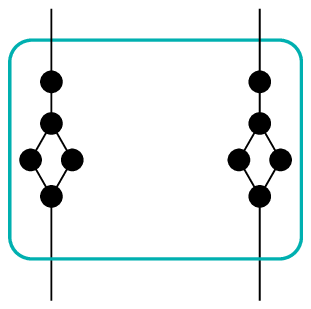}
  \label{gadget5}
}
\subfigure[Gadget 6]{
  \includegraphics[scale=.6]{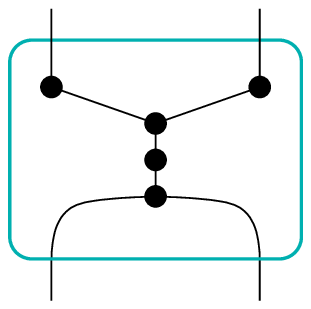} 
  \label{gadget6}
}
\subfigure[Gadget 7]{
  \includegraphics[scale=.6]{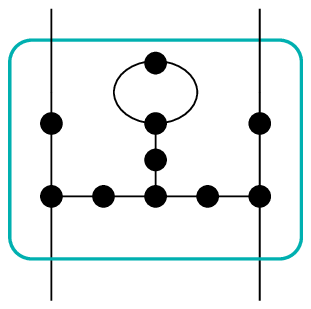}
  \label{gadget7}
}
\subfigure[Gadget 8]{
  \includegraphics[scale=.6]{figs/gadget8.ps}
  \label{gadget8}
}
\subfigure[Gadget 9]{ 
  \includegraphics[scale=.6]{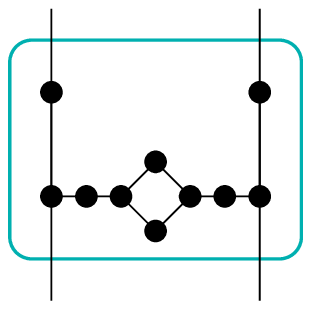} 
  \label{gadget9}
}
\caption{Binary recursive gadgets}
\label{binaryGadgets}
\end{figure}

%%%%%%%%% Begin Complex-Signatures %%%%%%%%%

\section{Complex Signatures} \label{complexSignatures}

Now we aim to characterize $\mathrm{Hol}(a,b)$
where $a, b \in \mathbb{C}$.  The next lemma introduces the technique
of algebraic symmetrization.
We show that over 3-regular graphs, the Holant value is expressible
as an integer polynomial $P(X, Y)$, where 
$X = ab$ and $Y = a^3 + b^3$. 
This change of variable, from $(a,b)$ to $(X, Y)$, is crucial in two ways.
First, it allows us to derive tractability results easily, drawing 
connections between problems that may appear unrelated, and the 
tractability of one implies the other.  Second, it facilitates the 
proof of hardness for those $(a,b)$ where the problem is indeed 
$\SHARPP$-hard by reducing the degree of the polynomials involved.
Once this transformation is made, 
four binary recursive gadgets easily 
cover all of the $\SHARPP$-hard problems where $X$ and $Y$ are real-valued, 
with a straightforward symbolic computation using 
{\sc CylindricalDecomposition} in Mathematica$^{\rm TM}$.
All gadget constructions in this section use
$[a,1,b]$ and $[1,0,0,1]$ signatures exclusively, and we
henceforth denote $X = ab$ and $Y = a^3 + b^3$ for the remainder of
this paper.

\begin{lemma}\label{transform}
Let $G$ be a 3-regular graph. Then there exists a polynomial 
$P(\cdot, \cdot)$ with two variables and
integer coefficients such that
for any signature grid $\Omega$ having underlying graph $G$
and every edge labeled $[a,1,b]$,
the Holant value is
$\mathrm{Holant}_\Omega = P(ab,a^3+b^3)$.
\end{lemma}

\proof
	  Consider any $\{0,1\}$ vertex assignment $\sigma$ with a non-zero
	valuation. If $\sigma'$ is the complement assignment switching all
	0's and 1's in $\sigma$, then for $\sigma$ and $\sigma'$,
	we have the sum of valuations 
	$a^i b^j + a^j b^i$ for some $i$ and $j$. Here $i$ (resp. $j$)
	is the number of
	edges connecting two degree 3 vertices  both assigned 0 (resp. 1) by
	$\sigma$. We note that $a^i b^j + a^j b^i = (ab)^{{\rm min}(i,j)}
	(a^{|i-j|}+ b^{|i-j|})$.
	
	We prove $i \equiv j \pmod 3$ inductively.
	For the all-0 assignment, this is clear since every edge
	contributes a factor $a$ and the number of edges is divisible
	by 3 for a 3-regular graph. Now starting from any
	assignment $\sigma$,  if we switch the assignment on one vertex
	from 0 to 1, it is easy to verify that it changes the valuation 
	from $a^i b^j$ to $a^{i'} b^{j'}$, where $i-j = i' - j' +3$.
	As every $\{0,1\}$ assignment is obtainable from the all-0 assignment
	by a sequence of switches, the conclusion $i \equiv j \pmod 3$ follows.
	
	Now $a^i b^j + a^j b^i = (ab)^{{\rm min}(i,j)} (a^{3k}+ b^{3k})$,
	for some $k \ge 0$ and a simple induction
	$a^{3(k+1)} + b^{3(k+1)}
	=(a^{3k}+ b^{3k})(a^3 + b^3) - (ab)^3 (a^{3(k-1)} + b^{3(k-1)})$
	shows that the Holant is a polynomial $P(ab, a^3 + b^3)$
	with integer coefficients.
\qed

\begin{corollary}\label{neatLittleTrick}
        If $X=-1$ and $Y \in \{0, \pm 2\mathfrak{i}\}$, then $\mathrm{Hol}(a,b)$ is 
        in $\ptime$.
\end{corollary}

\proof
	The problems $\mathrm{Hol}(1,-1)$, $\mathrm{Hol}(-\mathfrak{i},-\mathfrak{i})$, and
	$\mathrm{Hol}(\mathfrak{i},\mathfrak{i})$ are all solvable in 
	$\ptime$ (these fall within the families ${\mathcal F}_1$, ${\mathcal F}_2$, and ${\mathcal F}_3$ 
	in~\cite{Homomorphisms}); $X = -1$ for each, whereas the value of
	$Y$ for these problems is $0$, $2\mathfrak{i}$, and $-2\mathfrak{i}$ respectively.
	Since the value of any 3-regular
	signature grid is completely determined by $X$, $Y$, and 
	the polynomial $P(\cdot, \cdot)$ (which in turn depends only on the 
	underlying graph $G$), any $a$ and $b$ such that $ab = -1$ and $a^3 + b^3 \in \{0, \pm 2\mathfrak{i}\}$
	(i.e. $ab=-1$ and $a^{12}=1$) is computable in polynomial time.
\qed

We now list all of the cases where $\mathrm{Hol}(a,b)$
is computable in polynomial time.

\begin{theorem}\label{tractable}
        If any of the following four conditions is true, then
$\mathrm{Hol}(a,b)$ is solvable in $\ptime$:
        \begin{enumerate}
                \item $X=1$,
                \item $X=Y=0$,
                \item $X=-1$ and $Y \in \{0, \pm 2\mathfrak{i}\}$
                \item $4X^3 = Y^2$ and the input is restricted to planar graphs.
        \end{enumerate}
\end{theorem}
\proof
        If $X = 1$ then the signature $[a,1,b]$ is degenerate and the Holant 
        can be computed in polynomial time.  If $X = Y = 0$ then $a = b = 0$, and a 2-coloring 
        algorithm can be employed on the edges.  If $X = -1$ and 
        $Y \in \{0, \pm 2\mathfrak{i}\}$ then we are done by Corollary \ref{neatLittleTrick}.  
        If we restrict the input to planar graphs and $4X^3 = Y^2$ (equivalently, $a^3 = b^3$), holographic
        algorithms can be applied~\cite{STOC07}.
\qed

\begin{figure}[h!tb]
\centering
\subfigure[Gadget 10]{
  \includegraphics[scale=.62]{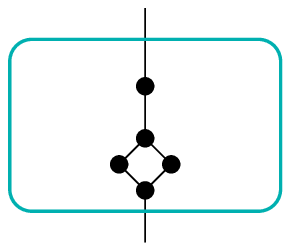}
  \label{gadget10}
}
\subfigure[Gadget 11]{
  \includegraphics[scale=.62]{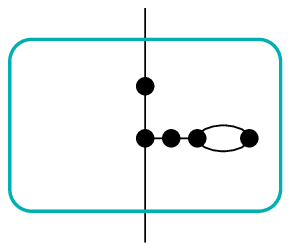}
  \label{gadget11}
}
\subfigure[Gadget 12]{
  \includegraphics[scale=.62]{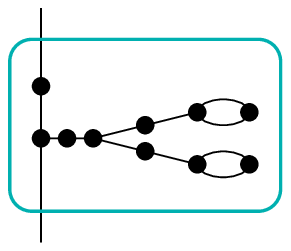}
  \label{gadget12}
}
\subfigure[Gadget 13]{
  \includegraphics[scale=.62]{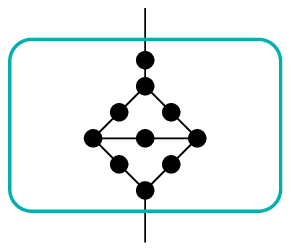}
  \label{gadget13}
}
\subfigure[Gadget 14]{
  \includegraphics[scale=.62]{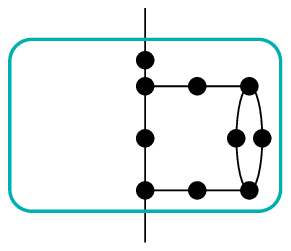}
  \label{gadget14}
}
\subfigure[Gadget 15]{
  \includegraphics[scale=.62]{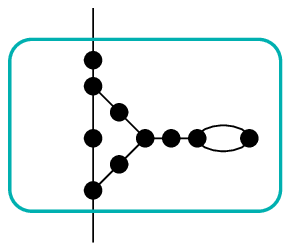}
  \label{gadget15}
}
\subfigure[Gadget 16]{
  \includegraphics[scale=.62]{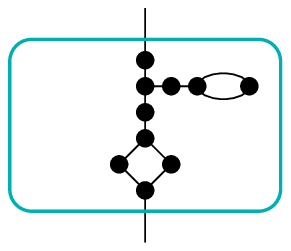}
  \label{gadget16}
}
\caption{Unary recursive gadgets}
\label{unaryGadgets}
\end{figure}

Our main task in this paper is to prove that
all remaining problems are $\SHARPP$-hard.  
The following two lemmas provide sufficient conditions 
to satisfy the eigenvalue 
requirement of the recursive constructions.

\begin{lemma} \label{relaxedUnary}
	If both roots of the complex polynomial $x^2 + Bx + C$ 
	have the same norm, then $B|C| = \overline{B}C$ and
	$B^2 \overline{C} = \overline{B}^2 C$.
	If further $B \ne 0$ and $C \ne 0$, then $\Arg(B^2) = \Arg(C)$.
\end{lemma}
\proof
	If the roots have equal norm, then for some $a,b\in \mathbb{C}$ and nonnegative $r\in \mathbb{R}$ and we can write $x^2 + Bx + C = (x-ra)(x-rb)$, where $|a| = |b| = 1$, so $B|C| = -r(a+b)r^2 = -r(a^{-1}+b^{-1}) r^2 ab = \overline{B}C$.  Squaring both sides and dividing by $C$, we have $B^2 \overline{C} = \overline{B}^2 C$ (note this is justified since this equality still holds when $C = 0$).  Multiplying $B|C| = \overline{B}C$ by $B$ we get $B^2|C| = |B^2|C$, and if $B$ and $C$ are both nonzero then $\frac{B^2}{|B^2|} = \frac{C}{|C|}$, that is, $\Arg(B^2) = \Arg(C)$.
\qed

\begin{lemma} \label{relaxedBinary}
	If all roots of the complex polynomial $x^3 + Bx^2 + Cx + D$
	have the same norm, then $C |C|^2 = \overline{B} |B|^2 D$.
\end{lemma}
\proof
	If the roots have equal norm, then for some $a,b,c\in \mathbb{C}$ and nonnegative $r\in \mathbb{R}$ we can write $x^3 + Bx^2 + Cx + D = (x-ra)(x-rb)(x-rc)$, where $|a| = |b| = |c| = 1$, so $B = -r(a+b+c)$, $C = r^2(ab+bc+ca)$, and $D = -r^3 abc$.  Then 
	\begin{eqnarray*}
		C |C|^2 = r^2(ab+bc+ca)r^4|ab+bc+ca|^2 = r(\overline{a+b+c})r^2|a+b+c|^2 r^3 abc = \overline{B} |B|^2 D,
	\end{eqnarray*}
	where we used the fact that $|ab+bc+ca| = |ab+bc+ca|\cdot|a^{-1}b^{-1}c^{-1}| = |a^{-1}+b^{-1}+c^{-1}| = |\overline{a+b+c}| = |a+b+c|$.
\qed

Now we introduce a powerful new technique called 
{\em Eigenvalue Shifted Pairs}.
\begin{definition}
        A pair of nonsingular square matrices $M$ and $M'$ is called an 
        {\em Eigenvalue Shifted Pair (ESP)} if $M' = M + \delta I$ for some non-zero 
        $\delta \in \mathbb{C}$, and $M$ has distinct eigenvalues.
\end{definition}
Clearly for such a pair, $M'$ also has distinct eigenvalues.
The recurrence matrices of Gadgets 10 and 11
(Figure \ref{unaryGadgets})
differ only by $ab-1$ along the diagonal, and
form an Eigenvalue Shifted Pair for nearly all $a,b \in \mathbb{C}$.
We will make significant use of 
such Eigenvalue Shifted Pairs, but first we state a technical lemma.

\begin{lemma}\label{eigenvalueShift}
	Suppose $\alpha, \beta, \delta \in \mathbb{C}$, $|\alpha| = |\beta|$,
	$\alpha \ne \beta$, $\delta \ne 0$, and 
	$|\alpha+\delta| = |\beta+\delta|$.  Then there exists 
	$r,s \in \mathbb{R}$ such that $r\delta = \alpha + \beta$ and 
	$s\delta^2 = \alpha \beta$.
\end{lemma}
\proof
	After a rotation in the complex plane, we can assume 
	$\alpha = \overline{\beta}$, and then since $\alpha + \beta, 
	\alpha \beta \in \mathbb{R}$ we just need to prove $\delta \in 
	\mathbb{R}$.  Then $(\alpha + \delta)\overline{(\alpha + \delta)} 
	= |\alpha + \delta|^2 = |\beta + \delta|^2 = (\beta + \delta)
	\overline{(\beta + \delta)} = (\overline{\alpha} + \delta)
	(\alpha + \overline{\delta})$ and we distribute to get $\alpha 
	\overline{\alpha} + \delta \overline{\delta} + \alpha 
	\overline{\delta} + \overline{\alpha} \delta = \alpha 
	\overline{\alpha} + \delta \overline{\delta} + 
	\overline{\alpha}\overline{\delta} + \alpha \delta$.  
	Canceling repeated terms and factoring, we have $(\overline{\alpha} 
	- \alpha)(\overline{\delta} - \delta) = 0$, and since $\alpha \ne 
	\beta = \overline{\alpha}$ we know $\overline{\delta} = \delta$ 
	therefore $\delta \in \mathbb{R}$.
\qed

\begin{corollary} \label{ESP}
        Let $M$ and $M'$ be an Eigenvalue Shifted Pair of $2$ by $2$ 
        matrices.  If both $M$ and $M'$ have eigenvalues of equal norm, then 
        there exists $r,s \in \mathbb{R}$ such that $\mathrm{tr}(M) = 
        r \delta$ (possibly $0$) and $\det(M) = s \delta^2$.
\end{corollary}
\proof
        Let $\alpha$ and $\beta$ be the eigenvalues of $M$, so
        $\alpha+\delta$ and $\beta+\delta$ are the eigenvalues of $M'$.  
        Suppose that $|\alpha| = |\beta|$ and $|\alpha + \delta| = 
        |\beta + \delta|$.  Then by Lemma \ref{eigenvalueShift}, there 
        exists $r,s \in \mathbb{R}$ such that $\mathrm{tr}(M) = \alpha 
        + \beta = r \delta$ and $\det(M) = \alpha \beta = s \delta^2$.
\qed

We now apply an ESP to prove that most settings of 
$\mathrm{Hol}(a,b)$ are $\SHARPP$-hard.
\begin{lemma} \label{condition0}
        Suppose $X \ne \pm 1$, $X^2+X+Y \ne 0$, 
        and $4(X-1)^2 (X+1) \ne (Y+2)^2$.  Then either unary Gadget 10 or unary Gadget 
        11 has nonzero eigenvalues with distinct norm, unless $X$ and $Y$ 
        are both real numbers.
\end{lemma}
\proof
        Gadgets 10 and 11 have
                $M_{10} = \left[\begin{array}{cc}
                        a^3 + 1& a + b^2 \\
                        a^2 + b & b^3 + 1
                       	\end{array}\right]$ and
                $M_{11} = \left[\begin{array}{cc}
                        a^3 + a b & a + b^2 \\
                        a^2 + b & a b + b^3
                        \end{array}\right]$
	as their recurrence matrices,
        so $M_{11} = M_{10} + (X - 1)I$, and the eigenvalue shift is nonzero.  
        Checking the determinants, $\det(M_{10}) = (X-1)^2 (X+1) \ne 0$ and 
        $\det(M_{11}) = (X-1)(X^2+X+Y) \ne 0$.  Also, 
        $\mathrm{tr}(M_{10})^2 - 4 \det(M_{10}) = (Y+2)^2 - 4 (X-1)^2 (X+1) \ne 0$,
        so the eigenvalues of $M_{10}$ are distinct.  Therefore by Corollary 
        \ref{ESP}, either $M_{10}$ or $M_{11}$ has nonzero eigenvalues of distinct 
        norm unless $\mathrm{tr}(M_{10}) = r (X-1)$ and $\det(M_{10}) = s (X-1)^2$ 
        for some $r, s \in \mathbb{R}$.  Then we would have $(X-1)^2 (X+1) = 
        s (X-1)^2$ so $X = s - 1 \in \mathbb{R}$ and $Y+2 = r (X-1)$ so $Y = 
        r(X-1)-2 \in \mathbb{R}$.
\qed
Now we will deal with the following exceptional cases from Lemma 
\ref{condition0} ($X =1$ is tractable by Theorem~\ref{tractable}).
\begin{enumerate}
        \item[0.] $X \in \mathbb{R}$ and $Y \in \mathbb{R}$
        \item[1.] $X^2 + X + Y = 0$
        \item[2.] $X = -1$
        \item[3.] $4(X-1)^2 (X+1) = (Y + 2)^2$
\end{enumerate}

The case where $X$ and $Y$ are both real 
is dealt with
using the tools developed in Section~\ref{interpolation},
 and some symbolic
computation.
This
includes the case where $a$ and $b$ are both
real as a subcase.
When $a$ and $b$ are both real,
a dichotomy theorem for the complexity of  $\mathrm{Hol}(a,b)$
has been proved in~\cite{TAMC} with a significant effort.
With the new tools developed, we offer a simpler proof.
This also covers some cases where $a$ or $b$ is complex.  Working
with real-valued $X$ and $Y$ is a significant advantage, since the failure 
condition given by Lemma \ref{relaxedBinary} is simplified by the 
disappearance of norms and conjugates.  This brings the problem of proving 
$\SHARPP$-hardness within reach of symbolic 
computation via cylindrical decomposition.  We apply Theorem \ref{mainBinary}
to Gadgets 4, 7, 8, and 9 (Figure \ref{binaryGadgets}) together with a
starter gadget (Figure \ref{starterGadget}) to prove that these problems are $\SHARPP$-hard.  
Conditions 1 and 2 of Theorem \ref{mainBinary} are encoded directly into 
a query for {\sc CylindricalDecomposition} in Mathematica$^{\rm TM}$, but 
first we give a lemma to show how to encode condition 2 of Lemma \ref{mainBinary}.

\begin{lemma} \label{starterCriterion}
	Suppose $M \in \mathbb{C}^{n \times n}$ and $s \in \mathbb{C}^{n \times 1}$.  If $\det([s,Ms,M^2 s,\dots,M^{n-1}s]) \ne 0$ then $s$ is not orthogonal to any row eigenvector of $M$.
\end{lemma}
\begin{proof}
	Suppose $s$ is orthogonal to a row eigenvector $v$ of $M$ with eigenvalue $\lambda$.  Then $v [s, Ms, ..., M^{n-1}s] = 0$, since $v M^i s = \lambda^i v s = 0$.  Since $v \not = 0$ this is a contradiction.
\end{proof}

\begin{theorem}\label{realCase}
	Suppose $a, b \in \mathbb{C}$, $X, Y \in \mathbb{R}$, 
	$X \ne 1$, $4X^3 \ne Y^2$, and it is not the case that both $X = -1$ and $Y = 0$.
	Then the problem $\mathrm{Hol}(a,b)$ is $\SHARPP$-hard.
\end{theorem}
\proof
	We will use binary recursive Gadgets 4, 7, 8, and 9 together with the single-vertex starter gadget given in Figure \ref{starterGadget} (denote the respective matrices by $M_4$, $M_7$, $M_8$, $M_{9}$, and $s$).  Calculating the recurrence matrices of these gadgets, we get
	\begin{eqnarray*}
		M_7 &=&
		\left[\begin{array}{ccc}
			a^6 + a^4 b + a^3 + a^2 b^2 & 2 a^4 + 4 a^2 b + 2 a b^3 & a^2 + a b^2 + b^4 + b \\
			a^5 + a^3 b + a^2 + a b^2 & a^4 b + a^3 + 2 a^2 b^2 + a b^4 + 2 a b + b^3 & a^2 b + a b^3 + b^5 + b^2 \\
			a^4 + a^2 b + a + b^2 & 2 a^3 b + 4 a b^2 + 2 b^4 & a^2 b^2 + a b^4 + b^6 + b^3
		\end{array}\right], \\
		M_8 &=& 
		\left[\begin{array}{ccc}
			a^6 + 2 a^3 + 1 & 2 a^4 + 4 a^2 b + 2 b^2 & a^2 + 2 a b^2 + b^4 \\
			a^5 + a^3 b + a^2 + b & 2 a^3 + 2 a^2 b^2 + 2 a b + 2 b^3 & a b^3 + a + b^5 + b^2 \\
			a^4 + 2 a^2 b + b^2 & 2 a^2 + 4 a b^2 + 2 b^4 & b^6 + 2 b^3 + 1
		\end{array}\right], \\
		M_9 &=& 
		\left[\begin{array}{ccc}
			a^6 + 2 a^3 + a^2 b^2 & 2 a^4 + 2 a^2 b + 2 a b^3 + 2 a & a^2 + b^4 + 2 b \\
			a^5 + 2 a^2 + a b^2 & a^4 b + a^3 + a^2 b^2 + a b^4 + 2 a b + b^3 + 1 & a^2 b + b^5 + 2 b^2 \\
			a^4 + 2 a + b^2 & 2 a^3 b + 2 a b^2 + 2 b^4 + 2 b & a^2 b^2 + b^6 + 2 b^3
		\end{array}\right].
	\end{eqnarray*}
	Calculating the characteristic polynomials $x^3 + B x^2 + C x + D$ of Gadgets 4, 7, 8, and 9, we get
	\begin{eqnarray*}
		B_4 &=& -(X + Y + 1) \\
		C_4 &=& (X^2 + X + Y)(X - 1) \\
		D_4 &=& -X(X-1)^3 \\
		B_{7} &=& -(- 2 X^3 + 4 X^2 + 2 X Y + 2 X + Y^2 + 2 Y) \\
		C_{7} &=& (X - 1) \cdot \\
		      && (X^5 - 4 X^4 - X^3 Y + 6 X^3 + 7 X^2 Y + 4 X^2 + 4 X Y^2 + 5 X Y + X + Y^3 + 2 Y^2 + Y) \\
		D_{7} &=& -(X - 1)^3 (2 X + Y) (X^4 - X^3 + X^2 Y + 3 X^2 + 2 X Y + X + Y^2 + Y ) \\
		B_{8} &=& -(- 2 X^3 + 2 X^2 + 2 X + Y^2 + 4 Y + 2) \\
		C_{8} &=& (X - 1)^2 (X^4 - 2 X^3 + 2 X^2 + 4 X Y + 6 X + 2 Y^2 + 4 Y + 1) \\
		D_{8} &=& -2 (X - 1)^6 X (X + 1) \\
		B_{9} &=& -(3 X^2 + X Y + 2 X + Y^2 + 3 Y + 1) \\
		C_{9} &=& (X - 1) \cdot \\
		      && (X^5 - 3 X^4 - 2 X^3 Y - X^3 + 4 X^2 Y + 7 X^2 + 2 X Y^2 + 6 X Y + 4 X + Y^3 + 4 Y^2 + 4 Y) \\
		D_{9} &=& -(X - 1)^3 (X + Y + 1) (X^4 - 2 X^3 + X^2 + 2 X Y + 4 X + Y^2 + 2 Y)
	\end{eqnarray*}
	Suppose $X \ne 1$, $4X^3 \ne Y^2$ (equivalently, $a^3 \ne b^3$), and it is not the case that both $X = -1$ and $Y = 0$.  For any real-valued setting of $X$ and $Y$ compatible with these constraints, we will see that at least one of these four binary recursive gadgets satisfies the requirements of Theorem \ref{mainBinary} (the only exception is $(X,Y) = (0,-1)$, but by Lemma \ref{transform} any such problem is equivalent to $\mathrm{Hol}(0, -1)$ which is known to be $\SHARPP$-hard \cite{TAMC, K09}).  To verify that Gadget $j$ satisfies condition 1 of Theorem \ref{mainBinary}, we apply Lemma \ref{relaxedBinary} and check that $D_{j} (B_{j}^3 D_{j} - C_{j}^3) \ne 0$ (note that the norm and conjugate disappear from the test since we are only considering real valued $X$ and $Y$).  By Lemma \ref{starterCriterion}, Gadget 4 satisfies condition 2 because $\det[s, M_4 s, M_4^2 s] = (X-1)^4 (b^3 - a^3) \ne 0$.  However, $\det[s, M_7 s, M_7^2 s] = (X-1)^5 (b^3 - a^3)(X^2 + X + Y)(X + Y + 1)$, $\det[s, M_{8} s, M_{8}^2 s] = (X-1)^5 (b^3 - a^3)(X^2 Y + 4X^2 + 2XY + Y^2 + Y)$, and $\det[s, M_{9} s, M_{9}^2 s] = (X - 1)^6 (b^3 - a^3) (X + 1) (Y + 2)$, so these are zero for some settings of $X$ and $Y$.  We summarize the essential observations in terms of $(X,Y)$ coordinates as follows.
	\begin{eqnarray*}
		X = 1 &\iff& ab = 1 \\
		4X^3 = Y^2 &\iff& a^3 = b^3 \\
		X = 0 \land Y = -1 &\iff& (a = 0 \land b^3 = -1) \lor \\
			&& \qquad(a^3 = -1 \land b = 0) \\
		X = -1 \land Y = 0 &\iff& a^6 = 1 \land ab = -1 \\
		D_{4} (B_4^3 D_4 - C_4^3) \ne 0 &\implies& \text{Gadget 4 fulfills Theorem \ref{mainBinary}} \\
		D_{7} (B_{7}^3 D_{7} - C_{7}^3)(X^2 + X + Y)(X + Y + 1) \ne 0 &\implies& \text{Gadget 7 fulfills Theorem \ref{mainBinary}} \\ 
		D_{8} (B_{8}^3 D_{8} - C_{8}^3)(X^2 Y + 4X^2 + 2XY + Y^2 + Y) \ne 0 &\implies& \text{Gadget 8 fulfills Theorem \ref{mainBinary}} \\ 
		D_{9} (B_{9}^3 D_{9} - C_{9}^3)(X + 1) (Y + 2) \ne 0 &\implies& \text{Gadget 9 fulfills Theorem \ref{mainBinary}} \\ 
	\end{eqnarray*} 
	%Note that $4X^3 = Y^2$ means $4a^3 b^3 = (a^3 + b^3)^2$, which is equivalent to $(a^3-b^3)^2 = 0$, so $4X^3 = Y^2$ if and only if $a^3 = b^3$. 
	If we can verify that at least one of the 8 conditions on the left hand side holds for any real-valued setting of $X$ and $Y$ then we are done.  Note that a disjunction of the left hand sides is a semi-algebraic set, and as such, is decidable by Tarski's Theorem \cite{Tarski}.  Using symbolic computation via the {\sc CylindricalDecomposition} function from Mathematica$^{\rm TM}$, we verify that for any $X, Y \in \mathbb{R}$, at least one of the eight conditions above is true.
\qed

Now we can assume that $X \notin \mathbb{R}$ or $Y \notin \mathbb{R}$,
and we deal with the remaining three conditions.
Note that if $X^2 + X + Y = 0$ then $X \in \mathbb{R}$ implies $Y \in \mathbb{R}$.  So in the following lemma, the assumption that $X$ and $Y$ are not both real numbers amounts to $X \notin \mathbb{R}$.
\begin{lemma} \label{condition2}
	If $X^2+X+Y=0$ and $X \notin \mathbb{R}$ then the recurrence matrix of
	unary Gadget 12 has nonzero eigenvalues with distinct norm.
\end{lemma}
\proof
	Let $M_{12}$ be the recurrence matrix for unary Gadget 12.  
	\begin{eqnarray*}
		M_{12} &=& 
		\left[\begin{array}{ccc}
			a^6 + 2 a^4 b + a^3 + 3 a^2 b^2 + a b^4 & a^4 + 3 a^2 b + 2 a b^3 + b^5 + b^2 \\
			a^5 + 2 a^3 b + a^2 + 3 a b^2 + b^4 & a^4 b + 3 a^2 b^2 + 2 a b^4 + b^6 + b^3
		\end{array}\right]
	\end{eqnarray*}
	Then the determinant is the polynomial
	$X^6 -6 X^5 - X^4 Y + 16 X^4 + 11 X^3 Y -10 X^3 + 5 X^2 Y^2 -7 X^2 Y - X^2 + X Y^3 -4 X Y^2 -3 X Y - Y^3 - Y^2$.
	{\em Amazingly}, with the condition $X^2+X+Y=0$, this polynomial factors into $-X^2(X-1)^5$.  Similarly, the trace, which is $-2 X^3 +6 X^2 +3 X Y + Y^2 + Y$, also factors into $X(X-1)^3$.  Since $\det(M_{12}) \ne 0$, $\mathrm{tr}(M_{12}) \ne 0$, and $(1-X)\det(M_{12}) = \mathrm{tr}(M_{12})^2$, we know $\Arg(\det(M_{12})) \ne \Arg(\mathrm{tr}(M_{12})^2)$ and conclude by Lemma \ref{relaxedUnary} that the eigenvalues of $M_{12}$ (which are nonzero) have distinct norm.
\qed
Similarly, Gadgets 11 and 13 can be used to deal with the $X = -1$ condition.
Recall that any setting of $a$ and $b$ such that $X=-1$ and $Y= \pm 2\mathfrak{i}$ 
is tractable by Theorem~\ref{tractable}.

\begin{lemma} \label{condition3}
	If $X=-1$, $Y \ne \pm 2\mathfrak{i}$, and $Y \notin \mathbb{R}$, then either 
	Gadget 11 or Gadget 13 has a recurrence
	matrix with nonzero eigenvalues with distinct norm.
\end{lemma}
\begin{proof}
	Suppose $|Y| \ne 2$, $Y \notin \mathbb{R}$, and let $M_{11}$ be the recurrence matrix for unary Gadget 11.  Well, $\det(M_{11}) = -2Y \ne 0$ and $\mathrm{tr}(M_{11}) = Y-2$, so $\overline{\mathrm{tr}(M_{11})}\cdot\det(M_{11}) - \mathrm{tr}(M_{11})\cdot|\det(M_{11})| = -(\overline{Y}-2)(2Y) - (Y-2)\cdot|-2Y| = 4Y - 2|Y|^2 - 2Y\cdot|Y| + 4|Y| = -2(|Y|-2)(|Y| + Y) \ne 0$.  Thus $\overline{\mathrm{tr}(M_{11})}\cdot\det(M_{11}) \ne \mathrm{tr}(M_{11})\cdot|\det(M_{11})|$ and by Lemma \ref{relaxedUnary}, $M_{11}$ has (nonzero) eigenvalues with distinct norm.

	Now suppose $|Y|=2$, but $Y \ne \pm 2\mathfrak{i}$ and $Y \notin \mathbb{R}$.  Let $M_{13}$ be the recurrence matrix for unary Gadget 13.  
	\begin{eqnarray*}
		M_{13} &=& 
		\left[\begin{array}{ccc}
			a^6 + 3 a^3 + 3 a b + b^3 & a^4 + 2 a^2 b + a b^3 + a + b^5 + 2 b^2 \\
			a^5 + a^3 b + 2 a^2 + 2 a b^2 + b^4 + b & a^3 + 3 a b + b^6 + 3 b^3
		\end{array}\right].
	\end{eqnarray*}
	Then $\det(M_{13}) = -16Y \ne 0$.  Using the substitution $\overline{Y} = 4/Y$,
	\begin{eqnarray*}
		\mathrm{tr}(M_{13})^2 \overline{\det(M_{13})} - {\qquad\qquad} &&\\
		\overline{\mathrm{tr}(M_{13})}^2 \det(M_{13})
			&=& -16 (Y - \overline{Y}) \cdot \\
			&& (-16 + 8 Y \overline{Y} + 8 Y^2 \overline{Y} + Y^3 \overline{Y} + 8 Y \overline{Y}^2 + Y^2 \overline{Y}^2 + Y \overline{Y}^3) \\
			&=& \frac{-64 (Y - \overline{Y}) (4+Y^2) (4+8 Y+Y^2)}{Y^2} \\
			&\ne& 0.
	\end{eqnarray*}
	Hence $\mathrm{tr}(M_{13})^2 \overline{\det(M_{13})} \ne \overline{\mathrm{tr}(M_{13})}^2 \det(M_{13})$ and the eigenvalues of $M_{13}$ (which are nonzero) have distinct norm by Lemma \ref{relaxedUnary}.
\end{proof}

The condition $4(X-1)^2 (X+1) = (Y+2)^2$ is somewhat resilient to individual
unary recursive gadgets, but by using a second Eigenvalue Shifted Pair, we 
can reduce it to simpler conditions.
\begin{lemma} \label{condition4Part1}
	Suppose $4(X-1)^2 (X+1) = (Y+2)^2$.  Then either unary Gadget 13 or 
	unary Gadget 14 has nonzero eigenvalues with distinct norm, unless
	either $X^3 + 2X^2 + X + 2Y = 0$, or $X^3 + 4X^2 + 2Y - 1 = 0$, or
	both $X, Y \in \mathbb{R}$.
\end{lemma}
\begin{proof}
	Assume that $X^3 + 2X^2 + X + 2Y \ne 0$, $X^3 + 4X^2 + 2Y - 1 \ne 0$,
	and it is not the case that both $X, Y \in \mathbb{R}$.  Note that
	$X \notin \{0, 1\}$ since otherwise $Y \in \mathbb{R}$ and we know 
	that $X$ and $Y$ are not both real.
	The recurrence matrix for Gadget 14 is
	\begin{eqnarray*}
		M_{14} &=& \left[\begin{array}{cc}
			a^6 + 3 a^3 + a^2 b^2 + a b + b^3 + 1 & a^4 + 2 a^2 b + a b^3 + a + b^5 + 2 b^2 \\
			a^5 + a^3 b + 2 a^2 + 2 a b^2 + b^4 + b & a^3 + a^2 b^2 + a b + b^6 + 3 b^3 + 1
		\end{array}\right]
	\end{eqnarray*}
	so $M_{14} = M_{13} + (X - 1)^2I$, and the eigenvalue shift is nonzero.  Now, $\det(M_{13}) = (X-1)^3 (X^3 + 2X^2 + X + 2Y) \ne 0$ and note that $\mathrm{tr}(M_{13}) = - 2X^3 + 6X + Y^2 + 4Y$ simplifies to $\mathrm{tr}(M_{13}) = - 2X^3 + 6X + Y^2 + 4Y - (Y+2)^2 + 4(X-1)^2 (X+1) = 2X(X-1)^2$ using the fact that $4(X-1)^2 (X+1) = (Y+2)^2$.
	
	Similarly, $\det(M_{14}) = \det(M_{14}) + (X-1)^2 (4(X-1)^2 (X+1) - (Y+2)^2) = (X-1)^3 (X^3 + 4X^2 + 2Y - 1) \ne 0$.  Furthermore $\mathrm{tr}[M_{13}]^2 - 4\det(M_{13}) = 4X^2(X-1)^4 - 4(X-1)^3 (X^3 + 2X^2 + X + 2Y) = -4(X-1)^3 (3 X^2 + X + 2Y)$.  If this is zero, then substituting $Y = (-3 X^2 - X)/2$ into $(Y+2)^2 - 4(X-1)^2 (X+1) = 0$ we get $X(X-1)^2(9X+8) = 0$ and $X \in \mathbb{R}$, with $Y \in \mathbb{R}$ as a direct consequence.  Corollary \ref{ESP} implies that either Gadget 13 or Gadget 14 has nonzero eigenvalues of distinct norm, unless $\mathrm{tr}(M_{13})=r (X - 1)^2$ and $\det(M_{13})=s (X - 1)^4$ for some $r,s \in \mathbb{R}$.  But then $2X(X-1)^2 = r (X-1)^2$ hence $X = r/2 \in \mathbb{R}$, and $(X-1)^3 (X^3 + 2X^2 + X + 2Y) = s (X - 1)^4$ hence $Y = (- X^3 - 2X^2 - X + s(X - 1))/2 \in \mathbb{R}$.  A contradiction.
\end{proof}

Now we take advantage of another interesting coincidence; two gadgets with
recurrence matrices that have identical trace.
\begin{lemma} \label{condition4Part2}
	If $X^2+X+Y \ne 0$, $4(X-1)^2 (X+1) = (Y+2)^2$, and either 
	$X^3 + 2X^2 + X + 2Y = 0$ or $X^3 + 4X^2 + 2Y - 1 = 0$, then the
	recurrence matrix of unary Gadget 15 or unary Gadget 16 has nonzero
	eigenvalues with distinct norm, unless both $X, Y \in \mathbb{R}$.
\end{lemma}
\begin{proof}
	The recurrence matrices for Gadget 15 and Gadget 16 are
	\begin{eqnarray*}
		M_{15} &=& \left[\begin{array}{cc}
			a^6 + a^4 b + 2 a^3 + a^2 b^2 + 2 a b + b^3 & a^4 + 3 a^2 b + 2 a b^3 + b^5 + b^2 \\
			a^5 + 2 a^3 b + a^2 + 3 a b^2 + b^4 & a^3 + a^2 b^2 + a b^4 + 2 a b + b^6 + 2 b^3
		\end{array}\right], \\
		M_{16} &=& \left[\begin{array}{cc}
			a^6 + a^4 b + 2 a^3 + a^2 b^2 + 2 a b + b^3 & a^4 + a^3 b^2 + a^2 b + 2 a b^3 + a + b^5 + b^2 \\
			a^5 + 2 a^3 b + a^2 b^3 + a^2 + a b^2 + b^4 + b & a^3 + a^2 b^2 + a b^4 + 2 a b + b^6 + 2 b^3
		\end{array}\right].
	\end{eqnarray*}
	Let $T = X^3 + 2X^2 + X + 2Y$, $U = X^3 + 4X^2 + 2Y - 1$, and let $R$ denote $(Y+2)^2 - 4(X-1)^2 (X+1)$.  Note that regardless of whether $T = 0$ or $U = 0$, $X \in \mathbb{R}$ implies $Y \in \mathbb{R}$, so we will assume $X \notin \mathbb{R}$.  The main diagonals of $M_{15}$ and $M_{16}$ are identical, so $\mathrm{tr}(M_{15}) = \mathrm{tr}(M_{16})$.  Furthermore, if $T = 0$ then $\mathrm{tr}(M_{15}) = \mathrm{tr}(M_{15}) - R - (X-1)T/2 = -X(X-1)^3/2 \ne 0$.  If $U = 0$ then $\mathrm{tr}(M_{15}) = \mathrm{tr}(M_{15}) - R - (X-1)U/2 = -(X-1)(X^3-1)/2$, and we claim this is nonzero as well.  Otherwise, $X^3 = 1$ then since $U = 0$, $Y = -2 X^2$ and using $(Y + 2)^2 = 4(X - 1)^2 (X + 1)$ we get $(X^2 - 1)^2 = (X - 1)^2 (X + 1)$ i.e. $(X - 1)^2 (X + 1)^2 = (X - 1)^2 (X + 1)$ together with $X \notin \mathbb{R}$ we get a contradiction.  Next, $\det(M_{16}) = (X-1)^3 (X+1)(X^2 + X + Y)$ and $\det(M_{15}) = \det(M_{15}) - R(X-1)^2 = (X-1)^3 (X+4)(X^2 + X + Y)$, so these are both nonzero.  If both $M_{15}$ and $M_{16}$ have eigenvalues with equal norm, then applying Lemma \ref{relaxedUnary} twice, $\Arg(\det(M_{15})) = \Arg(\mathrm{tr}(M_{15})^2) = \Arg(\mathrm{tr}(M_{16})^2) = \Arg(\det(M_{16}))$.  However, this would imply $\Arg(X+4) = \Arg(X+1)$ and $X \in \mathbb{R}$, so we conclude that either $M_{15}$ or $M_{16}$ has nonzero eigenvalues with distinct norm.
\end{proof}

Now we sum up the result of these lemmas.
\begin{theorem}
	Suppose $a, b \in \mathbb{C}$ such that $X \ne 1$, $4X^3 \ne Y^2$,
	and $(X,Y) \ne (-1,0)$.
	Then the problem $\mathrm{Hol}(a,b)$ is $\SHARPP$-hard.
\end{theorem}
\begin{proof}
	Under these assumptions, if $X$ and $Y$ are both real then $\mathrm{Hol}(a,b)$ is $\SHARPP$-hard by Lemma \ref{realCase}, so assume either $X$ or $Y$ is not real.  For any such $a$ and $b$, we know by Lemma \ref{condition0} that either Gadget 10 or 11 has a recurrence matrix with nonzero eigenvalues of distinct norm, except in the following cases, where we will use other gadgets to fill this requirement.
	\begin{enumerate}
		\item $X^2 + X + Y = 0$.
		\item $X = -1$.
		\item $4(X-1)^2 (X+1) = (Y + 2)^2$.
	\end{enumerate}
	If $X^2+X+Y = 0$ then $X \notin \mathbb{R}$, lest $X$ and $Y$ be real, so Lemma \ref{condition2} implies that unary Gadget 12 has a recurrence matrix of the required form.
	If $X=-1$, then Lemma \ref{condition3} indicates that either Gadget 11 or Gadget 13 satisfies the requirement, unless $Y = \pm 2\mathfrak{i}$.
	Now we may assume $X^2+X+Y \ne 0$, so by Lemmas \ref{condition4Part1} and \ref{condition4Part2} if $4(X-1)^2 (X+1) = (Y+2)^2$ then either unary Gadget 13, 14, 15, or 16 has a suitable recurrence matrix.
	In any case, we have a unary recursive gadget whose recurrence matrix has nonzero eigenvalues of distinct norm.  Hence we are done by Theorem \ref{mainUnary} and Lemma \ref{mainLemma2}.
\end{proof}

Recall {\sc Vertex Cover} is $\SHARPP$-hard on 3-regular planar graphs,
and note that all gadgets discussed are planar (in the case of
Gadget 8, each iteration can be redrawn in a planar way
by ``going around'' the previous iterations; 
see Figure \ref{iteration}).
Thus, all of the hardness results
proved so far still apply even when the input graphs are restricted to planar
graphs.  There are, however, some problems that are
$\SHARPP$-hard in general, yet polynomial time computable when the input is
restricted to planar graphs.  This class of problems corresponds
exactly with the problems we still need to resolve at this point, i.e. 
when $4X^2 = Y^3$ but $X \notin \{0, \pm 1\}$.
The relevant interpolation results can be obtained 
with Gadget 4 and holographic reductions, using a technique demonstrated in \cite{FOCS08}.
\begin{lemma} \label{abEqual}
	The problem $\#[a,1,a] \mid [1,0,0,1]$ is $\SHARPP$-hard, unless $a \in \{0,\pm 1, \pm\mathfrak{i}\}$, in which case it is in P.
\end{lemma}
\begin{proof}
	If $a \in \mathbb{R}$ then this is already known \cite{TAMC}, and a polynomial time algorithm for $a = \pm \mathfrak{i}$ is in \cite{Homomorphisms}.  Now assume $a \notin \mathbb{R}$ and $a \ne \pm \mathfrak{i}$.  Since these problems have an extra degree of symmetry, we use a 2 by 2 recurrence matrix to describe the recursive construction which consists of a single-vertex starter gadget (Figure \ref{starterGadget}) followed by some number of applications of binary recursive Gadget 4 (no finisher gadget is used here).  That is, if $\mathcal{F}$-gate $N_i$ has signature $[a_i, b_i, a_i]$, then the signature of $N_{i+1}$ is given by $[a_{i+1}, b_{i+1}, a_{i+1}]$ where $[a_{i+1}, b_{i+1}]^\mathrm{T} = M [a_i, b_i]^\mathrm{T}$, and $M = \left[\begin{array}{cc} a(a^2+1) & 2a \\ 2a^2 & a^2+1 \end{array}\right]$.  Now, $\det(M) = a(a-1)^2 (a+1)^2$ and $\mathrm{tr}(M) = (a+1)(a^2+1)$ are both nonzero under our assumptions.  It can be verified (using the {\sc Resolve} function of Mathematica$^{\rm TM}$) that $\mathrm{tr}(M) |\det(M)| \ne \overline{\mathrm{tr}(M)} \det(M)$ provided that $a \notin \mathbb{R}$, so by Lemma \ref{relaxedUnary}, the eigenvalues of $M$ have distinct norm.  Also, $M \left[\begin{array}{c} a \\ 1 \end{array}\right] = (a+1)\left[\begin{array}{c} a(a^2-a+2) \\ (2a^2-a+1) \end{array}\right]$, so the starter gadget is not an eigenvector of $M$.  We conclude by an analogous version of Lemma \ref{unaryConstruction} that we can interpolate all signatures of the form $[c,1,c]$, and since $\#[0,1,0] \mid [0,1,1,0]$ is known to be $\SHARPP$-hard \cite{FOCS08} and equivalent to $\#[-\omega(\omega^2-2)/3, 1, -\omega(\omega^2-2)/3] \mid [1,0,0,1]$ by holographic reduction (where $\omega$ is the principal 12$^\mathrm{th}$ root of unity), we conclude that $\#[a,1,a] \mid [1,0,0,1]$ is $\SHARPP$-hard.
\end{proof}

\begin{lemma}
	If $4X^3 = Y^2$, then $\mathrm{Hol}(a,b)$ is $\SHARPP$-hard unless $X \in \{0, \pm 1\}$, in which case it is in P.
\end{lemma}
\begin{proof}
	If $X = 0$ then $X = Y = 0$ and the problem is in P by Theorem \ref{tractable}.  Otherwise, $X \ne 0$, let $\omega = ba^{-1}$, and applying a holographic reduction to $\#[a,1,b] \mid [1,0,0,1]$ under the basis $\left[\begin{array}{cc} \omega & 0 \\ 0 & \omega^2 \end{array}\right]$ we see that the problem $\#[a, 1, b] \mid [1,0,0,1]$ is equivalent to $\#[\omega^2 a, 1, \omega b] \mid [1,0,0,1]$, because $\omega^3 = b^3 a^{-3} = 1$.  Since $\omega^2 a = \omega b$, we can apply Lemma \ref{abEqual} and the problem $\#[a,1,b] \mid [1,0,0,1]$ is in P if $ab = \omega^2 a \cdot \omega b = \pm 1$ and $\SHARPP$-hard otherwise.
\end{proof}

Given this, we have the following result.
\begin{theorem}
	The problem $\mathrm{Hol}(a,b)$ is $\SHARPP$-hard for all
	$a, b \in \mathbb{C}$ except in the following cases, for which 
	the problem is in $\ptime$.
	\begin{enumerate}
		\item $X=1$
		\item $X=Y=0$
		\item $X = -1$ and $Y \in \{0, \pm 2 \mathfrak{i} \}$
	\end{enumerate}
	If we restrict the input to planar graphs, then these three 
	categories are tractable in $\ptime$, as well as a fourth category 
	$4 X^3 = Y^2$, and the problem remains $\SHARPP$-hard in all other
	cases.
\qed
\end{theorem}
A simple coordinate change from $(X,Y)$ to 
$(X, (\frac{Y}{2})^2)$ translates this into Theorem \ref{thm:main-intro}.
%\begin{proof}
%	Immediate from the following.
%	\begin{enumerate}
%		\item $X = 1 \iff ab = 1$
%		\item $X = Z = 0 \iff a = b = 0$
%		\item $X = -1 \land Z = 0\iff 
%		ab = -1 \land (a^3 + b^3)^2 = 0 \iff 
%		ab = -1 \land a^3 + b^3 = 0 \iff 
%		b = -a^{-1} \land a^3 - a^{-3} = 0 \iff 
%		b = -a^{-1} \land a^6 = 1$
%		\item $X = -1 \land Z = -1 \iff 
%		ab = -1 \land (a^3 + b^3)^2 = -4 \iff
%		ab = -1 \land a^3 + b^3 = \pm 2\mathfrak{i} \iff
%		b = -a^{-1} \land a^6 \pm 2 \mathfrak{i} a^3 - 1 = 0 \iff
%		b = -a^{-1} \land a^3 = \pm \mathfrak{i} \iff 
%		b = -a^{-1} \land a^6 = -1$
%		\item $X^3 = Z \iff 4 a^3 b^3 = (a^3 + b^3)^2 \iff
%		(a^3-b^3)^2 = 0 \iff a^3 = b^3$
%	\end{enumerate}
%\end{proof}

%\newpage

%%%%%%%%%% Begin Bib %%%%%%%%%

\end{document}